\documentclass{article}

\newcommand{\sv}[1]{}%
\newcommand{\lv}[1]{#1}%
\usepackage{etoolbox}
\usepackage{tikz}

\usepackage[T1]{fontenc}
\usepackage[utf8]{inputenc}
\usepackage{graphicx} %
\usepackage{hyperref}
\usepackage{amsmath,amssymb,amsthm,mathtools,amsfonts}
\usepackage[capitalize,noabbrev]{cleveref}
\usepackage{xcolor}
\lv{\usepackage{fullpage}}
\usepackage{complexity}

\tikzset{
    standard edge/.style={very thick, line join=round},
    red edge/.style={standard edge, red},
    blue edge/.style={standard edge, blue},
    black edge/.style={standard edge},
    vertex/.style={circle, draw, fill=black, inner sep=0pt, minimum size=5pt}
}

\newtheorem{theorem}{Theorem} 
\newtheorem{lemma}[theorem]{Lemma}
\newtheorem{problem}[theorem]{Problem}
\crefname{problem}{Problem}{Problems}

\newcommand{\Vars}{\mathsf{Vars}}
\newcommand{\Clauses}{\mathsf{Clauses}}

\newcommand{\problembox}[4]{%
\begin{center}
\begingroup
\setlength{\fboxsep}{8pt}%
\setlength{\fboxrule}{0.6pt}%
\fbox{%
\begin{minipage}{0.90\linewidth}
  \begin{problem}[#1]\label[problem]{#2}
\leavevmode\par\vspace{0.6em}
\noindent\textbf{Input:}\quad\quad\, #3

\vspace{0.3em}
\noindent\textbf{Question:}\quad #4
\end{problem}
\end{minipage}}
\endgroup
\end{center}
}

\title{Edge-decomposition into Two Triangular Forests is \NP-complete\thanks{Both authors were supported by Polish National Science Centre SONATA-17 grant number 2021/43/D/ST6/03312.}}
\lv{%
\usepackage{authblk}
  \author[1]{Beniamin Bibrowski}%
  \author[1]{Tomáš Masařík}%

  \affil[1]{Institute of Informatics, Faculty of Mathematics, Informatics and Mechanics, University~of~Warsaw, Warszawa, Poland 

 \texttt{b.bibrowski@student.uw.edu.pl}
  \texttt{masarik@mimuw.edu.pl}}
}
\date{}

\begin{document}

\maketitle

\begin{abstract}
Let $\mathcal F$ be a graph class that is closed under topological minors and 1-sums, has decidable membership, contains a triangle, and is not the class of all graphs.
Recently, Lee, Liu, and Tsai [ICALP 2026] showed that the edge-decomposition problem into $k \geq 3$ elements of $\mathcal F$ is \NP-hard.
In particular, their general hardness reduction covers a long-standing problem on outerthickness (when $\mathcal F$ is the class of outerplanar graphs).
On the other hand, it is well known that decomposing a graph into forests is polynomial-time solvable, as implied by work of Edmonds [J.\ Res.\ Natl.\ Bur.\ Stand.\ B.\ 1965].

In this paper, we take a first step toward determining the complexity of edge-decomposition problems into just two graphs (the case $k=2$).
We consider the simplest possible graph class $\mathcal F$ satisfying the criteria above: the \emph{triangular forests}, that is, graphs in which every 2-connected component is a triangle.
We prove that determining whether a graph can be edge-decomposed into two triangular forests is \NP-complete.

\begin{center}
\textbf{keywords:} \NP-hardness, triangular forest, edge-decomposition problems
\end{center}
\end{abstract}

\section{Introduction}

Edge-decomposition of a graph into the minimum number of planar graphs (known as graph \emph{thickness}) was among the first problems to be shown \NP-complete by an elaborate reduction from a variant of the 3-SAT problem~\cite{Mansfield1983}.
The lesser-known problem of determining the complexity of edge-decomposing a graph into the minimum number of outerplanar graphs (graph \emph{outerthickness}) remained open for many years and was settled only recently~\cite{outerthickness}.
In fact, Lee, Liu, and Tsai~\cite{outerthickness} show \NP-completeness of a much broader set of edge-decomposition problems.
For every positive integer $n$, we write $[n]\coloneqq\{1,\ldots,n\}$.
We define the general \emph{$\mathcal F$-edge-decomposition} decision problem formally as follows:
\problembox
  {$\mathcal F$-Edge-Decomposition}
  {prob:Fdecomp}
  {An undirected simple graph $G=(V,E)$ and an integer $k\geq 1$.}
  {Is there a partition $E=E_1\cup \cdots \cup E_k$ such that $(V,E_i)\in \mathcal F$ for every $i\in[k]$?}
Note that in \cref{prob:Fdecomp} empty parts in the decomposition are allowed.
In \cite{outerthickness}, the problem was formulated as an edge-covering problem.
However, for monotone classes, and in particular for classes closed under topological minors, the covering and partitioning formulations are equivalent.

The main theorem of~\cite{outerthickness}, together with the subsequent membership observation, states the following.

\pagebreak[2]
\begin{theorem}[{Lee, Liu, and Tsai~\cite[Theorem~1.2]{outerthickness}}]
\label{thm:llt}
Let $\mathcal F$ be a proper graph class.
Suppose that $\mathcal F$:
\begin{itemize}
\item has decidable membership,
\item is closed under topological minors,
\item is closed under 1-sums, and
\item contains a cycle of length $3$.
\end{itemize}
Then, for every fixed integer $k\geq 3$, $\mathcal F$-\textsc{Edge-Decomposition} is \NP-hard.
Moreover, if membership in $\mathcal F$ can be decided in polynomial time, then $\mathcal F$-\textsc{Edge-Decomposition} is \NP-complete.
\end{theorem}

In this paper, we take a first step toward extending \cref{thm:llt} to the case of $k=2$.
In fact, we show \NP-completeness for the simplest possible graph class $\mathcal F$ satisfying the properties imposed by \cref{thm:llt}.
Indeed, any such class must contain $K_3$; closure under 1-sums then forces it to contain all graphs obtained by gluing copies of $K_3$ at cutvertices, and closure under topological minors forces it to contain all topological minors of those graphs.
The resulting graph class is exactly the class of \emph{triangular forests}.
These are graphs in which every 2-connected component is a triangle, equivalently graphs with no cycle of length at least four.
Thus triangular forests form the smallest graph class, under inclusion, satisfying the requirements imposed by \cref{thm:llt}.

\problembox
  {\textsc{Triangular-Forest Decomposition}}
  {prob:TF}
  {An undirected simple graph $G=(V,E)$ and an integer $k\geq 1$.}
  {Is there a partition $E=E_1\cup \cdots \cup E_k$ such that $(V,E_i)$ is a triangular forest for every $i\in[k]$?}

We prove the following.
\begin{theorem}
\label{thm:mainGeneral}
For every fixed integer $k\ge 2$, \textsc{Triangular-Forest Decomposition} is \NP-complete.
\end{theorem}

Triangular forests have also appeared in the context of vertex-partitions.
Knauer, Rambaud, and Ueckerdt~\cite{KnauerRambaudUeckerdt2024} proved that every planar graph admits a vertex-partition into two sets, each inducing a triangular forest; they also point out that this result follows from earlier work of Thomassen~\cite{Thomassen1995}.

It is well known that edge-decompositions into forests are polynomial-time solvable~\cite{NashWilliams64}.
In fact, edge-decompositions into pseudoforests are still polynomial-time solvable.
A \emph{pseudoforest} is a graph in which every connected component contains at most one cycle, and the corresponding edge-decomposition problem asks whether the edge set of the input graph can be partitioned into $k$ pseudoforests.
These results are also implied by the matroid-partition theorem of Edmonds~\cite{Edmonds1965}.

However, for more restrictive graph classes, the decomposition problem may become hard.
For example, the \emph{star arboricity} problem asks for an edge-partition into \emph{star forests} (forests of star graphs).
The decision version asking whether an edge-decomposition into two star forests exists is \NP-complete by a result of Hakimi, Mitchem, and Schmeichel~\cite{starForestHard}.
Similarly, the \emph{linear arboricity} problem asks for an edge-partition into \emph{linear forests} (forests whose connected components are paths).
P\'eroche~\cite{Peroche82} showed that deciding whether the linear arboricity of a graph is two is \NP-complete, even for graphs of maximum degree four.
More specifically, Campbell, H\"orsch, and Moore~\cite{CHM24} studied decompositions into two linear forests with bounded component lengths, proving several \NP-completeness results and one polynomial-time case.
Here, a $k$-bounded linear forest is a linear forest in which every connected component is a path with at most $k$ edges.
This direction was further extended by Banerjee, Marciano, Mond, Petr, and Portier~\cite{BMMPP25}, who showed that deciding whether a graph can be edge-decomposed into a matching and a $k$-bounded linear forest is \NP-complete for every $k\ge 3$.

There are also related partitioning problems in which connectivity of parts is prescribed.
P\'alv\"olgyi~\cite{Palvolgyi10} proved that deciding whether a simple graph can be edge-partitioned into two trees is \NP-complete.
Bern\'ath and Kir\'aly~\cite{BK15} studied a systematic family of packing, covering, and partitioning problems involving paths, forests, trees, and spanning trees.
For instance, they showed that edge-partitioning into a spanning tree and a tree is \NP-complete, while edge-partitioning into a spanning tree and a forest is polynomial-time solvable.
For decompositions into paths, deciding whether a graph of maximum degree four admits an edge-decomposition into two paths is \NP-hard by a result of P\'eroche~\cite{Peroche84}, while determining the number of paths in the decomposition becomes polynomial-time solvable on cubic graphs~\cite{MWY26}.

Another classical line of research considers edge-partitions into fixed subgraphs.
A central result in this direction is due to Holyer~\cite{Holyer1981}.
He proved that, for every fixed $n\ge 3$, deciding whether a graph can be edge-partitioned into copies of $K_n$ is \NP-complete.

\subsection{Our Result}

As discussed, the cases $k\ge3$ are implied by \cref{thm:llt}.
Hence, we focus on the case $k=2$.
Moreover, membership in \NP\ is trivial, as one can easily check in polynomial time that each graph in the partition is indeed a triangular forest.
Lastly, we can impose a degree restriction on the graphs in our construction.
Hence, by proving the following theorem, we conclude~\cref{thm:mainGeneral}.

\begin{theorem}\label{thm:main}
Even when restricted to graphs of maximum degree at most $9$, it is \NP-hard to decide whether a given graph can be edge-decomposed into two triangular forests.
\end{theorem}

We prove our theorem by a reduction from the \textsc{Monotone NAE-4-SAT} problem, which was shown to be \NP-hard by Schaefer~\cite{Schaefer1978}; see also Phelps and Rödl~\cite{PR} for the hardness in the equivalent language of 2-colorability of 4-uniform hypergraphs.

\problembox
  {\textsc{Monotone NAE-4-SAT} (Monotone Not-All-Equal-4-SAT)}
  {prob:MNAE4SAT}
  {A formula that consists of a set $\Vars$ of Boolean variables and a set $\Clauses$ of clauses, each composed of exactly four distinct variables from $\Vars$.}
  {Is there a valuation $\sigma\colon \Vars\to\{\top,\bot\}$ such that, for every clause $C\in\Clauses$, at least one variable of $C$ is assigned $\top$ and at least one variable of $C$ is assigned $\bot$?}

We regard $\Clauses$ as a set. If an input encoding contains duplicate clauses, they can be deleted without changing whether the formula is NAE-satisfiable.

\section{Proof of \cref{thm:main}}

In this section, we interpret a decomposition into two triangular forests as a Red/Blue edge-coloring such that neither monochromatic subgraph contains a cycle of length at least $4$.
We call such an edge-coloring \emph{admissible}.

To build our reduction, we first analyze the fundamental properties of a $K_5$ clique decomposed into two triangular forests as depicted in \cref{fig:k5_bull_decomposition}.
Before analyzing the decomposition, we formally define the bull graph. The \emph{bull graph} is a graph on five vertices consisting of a central triangle and two pendant edges attached to two distinct vertices of the triangle. We refer to these two disjoint pendant edges as the \emph{horns} of the bull graph; see~\cref{fig:bull}.

\begin{figure}[t]
    \centering
    \begin{tikzpicture}
        \node[vertex] (v4) at (-0.5, 1) {};
        \node[vertex] (v5) at (2.5, 1) {};
        \node[vertex] (v1) at (0, 0) {};
        \node[vertex] (v2) at (2, 0) {};
        \node[vertex] (v3) at (1, -1.5) {};

        \draw[standard edge] (v4) -- (v1); %
        \draw[standard edge] (v5) -- (v2); %
        \draw[standard edge] (v1) -- (v2) -- (v3) -- (v1); %
    \end{tikzpicture}
    \caption{The bull graph, consisting of a triangle and two disjoint pendant edges, called \emph{horns}, attached to distinct vertices.}
    \label{fig:bull}
\end{figure}
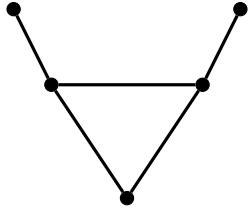
\begin{lemma}[Bull Decomposition of $K_5$]\label[lemma]{lem:bull}
Both color classes in an admissible coloring of a $K_5$ must form a connected spanning bull subgraph (see \cref{fig:bull}).
Moreover, every Red horn shares exactly one vertex with every Blue horn.
\end{lemma}
The proof of \cref{lem:bull} closely follows a classical result in Ramsey theory, established by Chartrand and Schuster~\cite{chartrand1971existence}.

\begin{proof}[Proof of \cref{lem:bull}]
In their proof of $R(C_4,C_4)=6$, Chartrand and Schuster~\cite{chartrand1971existence} showed that every graph on $5$ vertices with $6$ or more edges contains a cycle of length at least $4$, unless it is the bowtie graph. 
In the exceptional case, its complement on $5$ vertices contains a $C_4$.
Thus, an admissible coloring must partition the $10$ edges of $K_5$ into two sets of exactly $5$ edges.

To classify the possible color classes, consider a graph on $5$ vertices with $5$ edges and no $C_{\ge 4}$.
It must be connected: a disconnected graph with $5$ edges would have a component on $4$ vertices with $5$ edges, which contains a $C_4$.
It is therefore unicyclic, and its unique cycle must be a triangle.
Distributing the two remaining edges among the two vertices outside this triangle gives, up to isomorphism, exactly three possibilities: a triangle with a path of length $2$ attached, a triangle with two pendant edges at the same vertex, and the bull graph. 
Out of these, only the bull graph has a $C_{\ge 4}$-free complement. Therefore, both color classes form spanning, connected bull graphs.

Let the Red bull have a central triangle $v_1v_2v_3$ and horns $v_1v_4$ and $v_2v_5$. Its complement, the Blue bull, necessarily has horns $v_1v_5$ and $v_2v_4$. The Red horn $v_1v_4$ shares exactly $v_1$ with $v_1v_5$ and exactly $v_4$ with $v_2v_4$.
By symmetry, every Red horn shares exactly one vertex with every Blue horn.
\end{proof}

Using this structure, we can construct a robust color-propagating gadget.

\begin{figure}[t]
    \centering
    \begin{tikzpicture}

        \node[vertex, label=90:$v_0$] (v0) at (90:2) {};
        \node[vertex, label=162:$v_1$] (v1) at (162:2) {};
        \node[vertex, label=234:$v_2$] (v2) at (234:2) {};
        \node[vertex, label=306:$v_3$] (v3) at (306:2) {};
        \node[vertex, label=18:$v_4$] (v4) at (18:2) {};

        \draw[red edge] (v0) -- (v1) -- (v4) -- (v0);
        \draw[red edge] (v1) -- (v2);
        \draw[red edge] (v4) -- (v3);

        \draw[blue edge] (v2) -- (v3) -- (v0) -- (v2);
        \draw[blue edge] (v2) -- (v4);
        \draw[blue edge] (v3) -- (v1);

    \end{tikzpicture}
    \caption{Edge-decomposition of a $K_5$ clique into a Red bull graph and a Blue bull graph, illustrating \cref{lem:bull}.}
    \label{fig:k5_bull_decomposition}
\end{figure}
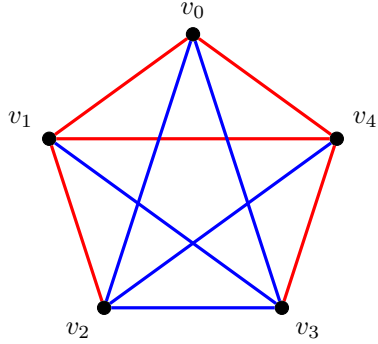

\subsection{Construction}

We describe the construction of the gadgets starting with an instance of \cref{prob:MNAE4SAT}.

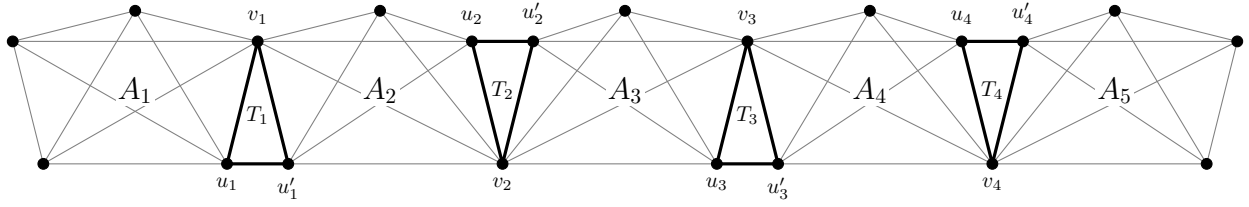
\begin{figure}[t]
    \centering
    \begin{tikzpicture}[
        scale=0.81, %
        transform shape
    ]

        \begin{scope}[gray, thin]
            \coordinate (A1_0) at (0, 1.5);
            \coordinate (A1_1) at (2, 1);      %
            \coordinate (A1_2) at (1.5, -1);   %
            \coordinate (A1_3) at (-1.5, -1);
            \coordinate (A1_4) at (-2, 1);
            
            \draw (A1_0)--(A1_1) (A1_0)--(A1_2) (A1_0)--(A1_3) (A1_0)--(A1_4);
            \draw (A1_1)--(A1_2) (A1_1)--(A1_3) (A1_1)--(A1_4);
            \draw (A1_2)--(A1_3) (A1_2)--(A1_4);
            \draw (A1_3)--(A1_4);

            \coordinate (A2_0) at (4, 1.5);
            \coordinate (A2_1) at (5.5, 1);    %
            \coordinate (A2_2) at (6, -1);     %
            \coordinate (A2_3) at (2.5, -1);   %
            \coordinate (A2_4) at (2, 1);      %
            
            \draw (A2_0)--(A2_1) (A2_0)--(A2_2) (A2_0)--(A2_3) (A2_0)--(A2_4);
            \draw (A2_1)--(A2_2) (A2_1)--(A2_3) (A2_1)--(A2_4);
            \draw (A2_2)--(A2_3) (A2_2)--(A2_4);
            \draw (A2_3)--(A2_4);

            \coordinate (A3_0) at (8, 1.5);
            \coordinate (A3_1) at (10, 1);     %
            \coordinate (A3_2) at (9.5, -1);   %
            \coordinate (A3_3) at (6, -1);     %
            \coordinate (A3_4) at (6.5, 1);    %
            
            \draw (A3_0)--(A3_1) (A3_0)--(A3_2) (A3_0)--(A3_3) (A3_0)--(A3_4);
            \draw (A3_1)--(A3_2) (A3_1)--(A3_3) (A3_1)--(A3_4);
            \draw (A3_2)--(A3_3) (A3_2)--(A3_4);
            \draw (A3_3)--(A3_4);

            \coordinate (A4_0) at (12, 1.5);
            \coordinate (A4_1) at (13.5, 1);   %
            \coordinate (A4_2) at (14, -1);    %
            \coordinate (A4_3) at (10.5, -1);  %
            \coordinate (A4_4) at (10, 1);     %
            
            \draw (A4_0)--(A4_1) (A4_0)--(A4_2) (A4_0)--(A4_3) (A4_0)--(A4_4);
            \draw (A4_1)--(A4_2) (A4_1)--(A4_3) (A4_1)--(A4_4);
            \draw (A4_2)--(A4_3) (A4_2)--(A4_4);
            \draw (A4_3)--(A4_4);

            \coordinate (A5_0) at (16, 1.5);
            \coordinate (A5_1) at (18, 1);
            \coordinate (A5_2) at (17.5, -1);
            \coordinate (A5_3) at (14, -1);    %
            \coordinate (A5_4) at (14.5, 1);   %
            
            \draw (A5_0)--(A5_1) (A5_0)--(A5_2) (A5_0)--(A5_3) (A5_0)--(A5_4);
            \draw (A5_1)--(A5_2) (A5_1)--(A5_3) (A5_1)--(A5_4);
            \draw (A5_2)--(A5_3) (A5_2)--(A5_4);
            \draw (A5_3)--(A5_4);
        \end{scope}

        \draw[black edge] (A1_1) -- (A1_2) -- (A2_3) -- cycle;

        \draw[black edge] (A2_2) -- (A2_1) -- (A3_4) -- cycle;

        \draw[black edge] (A3_1) -- (A3_2) -- (A4_3) -- cycle;

        \draw[black edge] (A4_2) -- (A4_1) -- (A5_4) -- cycle;

        \node[vertex] at (A1_0) {}; 
        \node[vertex, label={[yshift=2pt]above:$v_1$}] at (A1_1) {};
        \node[vertex, label=below:$u_1$] at (A1_2) {}; 
        \node[vertex] at (A1_3) {}; 
        \node[vertex] at (A1_4) {};

        \node[vertex] at (A2_0) {}; 
        \node[vertex, label={[yshift=2pt]above:$u_2$}] at (A2_1) {};
        \node[vertex, label=below:$v_2$] at (A2_2) {}; 
        \node[vertex, label=below:$u'_1$] at (A2_3) {};

        \node[vertex] at (A3_0) {}; 
        \node[vertex, label={[yshift=2pt]above:$v_3$}] at (A3_1) {};
        \node[vertex, label=below:$u_3$] at (A3_2) {}; 
        \node[vertex, label={[yshift=2pt]above:$u'_2$}] at (A3_4) {};

        \node[vertex] at (A4_0) {}; 
        \node[vertex, label={[yshift=2pt]above:$u_4$}] at (A4_1) {};
        \node[vertex, label=below:$v_4$] at (A4_2) {}; 
        \node[vertex, label=below:$u'_3$] at (A4_3) {};

        \node[vertex] at (A5_0) {}; 
        \node[vertex] at (A5_1) {};
        \node[vertex] at (A5_2) {}; 
        \node[vertex, label={[yshift=2pt]above:$u'_4$}] at (A5_4) {};

        \node[fill=white, inner sep=1pt] at (0, 0.2) {\Large $A_1$};
        \node[fill=white, inner sep=1pt] at (4, 0.2) {\Large $A_2$};
        \node[fill=white, inner sep=1pt] at (8, 0.2) {\Large $A_3$};
        \node[fill=white, inner sep=1pt] at (12, 0.2) {\Large $A_4$};
        \node[fill=white, inner sep=1pt] at (16, 0.2) {\Large $A_5$};

        \node[fill=white, inner sep=1pt] at (2, -0.2) {$T_1$};
        \node[fill=white, inner sep=1pt] at (6, 0.2) {$T_2$};
        \node[fill=white, inner sep=1pt] at (10, -0.2) {$T_3$};
        \node[fill=white, inner sep=1pt] at (14, 0.2) {$T_4$};

    \end{tikzpicture}
    \caption{The core part of a variable gadget of length $1$. The gadget propagates its color through a sequence of $K_5$ cliques ($A_1, A_2, A_3, A_4, A_5$) chained together by inside triangles ($T_1, T_2, T_3, T_4$), which must all be monochromatic by \cref{lem:variable}.}
    \label{fig:var}
    \end{figure}

\medskip
\noindent\textbf{Variable Gadget.~~} 
Informally, a variable gadget consists of copies of $K_5$ arranged in a path and connected by inside triangles, as depicted in \cref{fig:var}.
Fix a variable $V$, let
\[
O_V\coloneqq\{(V,C):C\in\Clauses\text{ and }V\in C\}
\]
be the set of occurrences of $V$, and set $\ell\coloneqq |O_V|$; we call $\ell$ the \emph{length} of the variable gadget.
We may assume that $\ell\geq 1$, since variables that do not occur in any clause can be deleted.
We fix an arbitrary bijection $\rho_V\colon[\ell]\to O_V$.
While describing this gadget, we suppress the superscript $V$ on its vertices, cliques, and triangles; we restore it when objects from several variable gadgets occur together.

The core part of the variable gadget begins with $4\ell+1$ copies of $K_5$, denoted by $A_1,\ldots,A_{4\ell+1}$.
For every $i\in[4\ell]$, the cliques $A_i$ and $A_{i+1}$ share exactly one vertex, denoted by $v_i$.
These are the only intersections among the cliques: nonconsecutive cliques are vertex-disjoint.
For every $i\in[4\ell]$, choose vertices
\[
u_i\in V(A_i)\setminus\{v_i\}
\qquad\text{and}\qquad
u'_i\in V(A_{i+1})\setminus\{v_i\}
\]
so that, for every $2\leq i\leq 4\ell$, the four vertices $v_{i-1},u'_{i-1},v_i,u_i$ are all distinct in $A_i$.
Such choices are always possible.
Indeed, $v_{i-1}\ne v_i$, since otherwise the nonconsecutive cliques $A_{i-1}$ and $A_{i+1}$ would intersect, and the $K_5$ clique $A_i$ has three further vertices from which distinct vertices $u'_{i-1}$ and $u_i$ can be chosen.
The pairwise-distinct condition ensures that the two inside-triangle edges used in $A_i$ are disjoint horns when required below. %

For every $i\in[4\ell]$, we add the edge $u_iu'_i$, which, together with the edges $u_iv_i\in E(A_i)$ and $v_iu'_i\in E(A_{i+1})$, forms the \emph{inside triangle} $T_i\coloneqq u_iv_iu'_i$.
For every $2\leq i\leq 4\ell$, define
\[
e_i\coloneqq u'_{i-1}v_i\in E(A_i).
\]
For each $i\in[\ell]$, the cliques $A_{4i-3},\ldots,A_{4i}$ form the $i$-th \emph{segment}, denoted by $S_i$, and this segment represents the occurrence $\rho_V(i)$.
We call the construction described so far the \emph{core part} of the variable gadget.

For every $i\in[\ell]$, introduce two new vertices $w_{2i-1}$ and $w_{2i}$.
We form the \emph{attachment triangle} $Q_{2i-1}$ by joining $w_{2i-1}$ to both endpoints of $e_{4i-2}$, and form $Q_{2i}$ analogously using $w_{2i}$ and $e_{4i}$.
We call $w_{2i-1}$ and $w_{2i}$ the two \emph{tips} of $S_i$.
Finally, we introduce a new vertex $z_i$ and add the three edges $w_{2i-1}w_{2i}$, $w_{2i-1}z_i$, $w_{2i}z_i$.
They form the \emph{central triangle} $P_i\coloneqq w_{2i-1}w_{2i}z_i$.
See~\cref{fig:var_with_clause}.
The edge $w_{2i-1}z_i$ is the \emph{clause edge} associated with the occurrence $\rho_V(i)$.

\begin{figure}[t]
    \begin{center}
    \hspace*{-0.2cm}
    \begin{tikzpicture}[
        scale=0.75
    ]
        \begin{scope}[gray, thin]
            \coordinate (A0) at (0, 1.5);
            \coordinate (A1) at (2, 1);
            \coordinate (A2) at (1.5, -1);
            \coordinate (A3) at (-1.5, -1);
            \coordinate (A4) at (-2, 1);
            
            \draw (A0)--(A1) (A0)--(A2) (A0)--(A3) (A0)--(A4);
            \draw (A1)--(A2) (A1)--(A3) (A1)--(A4);
            \draw (A2)--(A3) (A2)--(A4);
            \draw (A3)--(A4);

            \coordinate (B0) at (4, 1.5);
            \coordinate (B1) at (5.5, 1);
            \coordinate (B2) at (6, -1);
            \coordinate (B3) at (2.5, -1);

            \draw (B0)--(B1) (B0)--(B2) (B0)--(B3) (B0)--(A1);
            \draw (B1)--(B2) (B1)--(B3) (B1)--(A1);
            \draw (B2)--(B3) (B2)--(A1);
            \draw (B3)--(A1);

            \coordinate (C0) at (8, 1.5);
            \coordinate (C1) at (9.5, 1);
            \coordinate (C2) at (10, -1);
            \coordinate (C4) at (6.5, 1);
            
            \draw (C0)--(C1) (C0)--(C2) (C0)--(B2) (C0)--(C4);
            \draw (C1)--(C2) (C1)--(B2) (C1)--(C4);
            \draw (C2)--(B2) (C2)--(C4);
            \draw (B2)--(C4);

            \coordinate (D0) at (12, 1.5);
            \coordinate (D1) at (13.5, 1);
            \coordinate (D2) at (14, -1);
            \coordinate (D3) at (11, -1);

            \draw (D0)--(D1) (D0)--(D2) (D0)--(D3) (D0)--(C1);
            \draw (D1)--(D2) (D1)--(D3) (D1)--(C1);
            \draw (D2)--(D3) (D2)--(C1);
            \draw (D3)--(C1);
        \end{scope}

        \draw[black edge] (A1) -- (A2) -- (B3) -- cycle;
        \draw[black edge] (B2) -- (B1) -- (C4) -- cycle;
        \draw[black edge] (C1) -- (C2) -- (D3) -- cycle;

        \coordinate (Y1) at (4.25, -4);
        \coordinate (Y2) at (12.5, -4);

        \draw[blue edge] (B2) -- node[above, text=black, font=\small, yshift=2pt] {$e_{4i-2}$} (B3);
        \draw[blue edge] (B2) -- (Y1) -- (B3);

        \draw[blue edge] (D2) -- node[above, text=black, font=\small, yshift=2pt] {$e_{4i}$} (D3);
        \draw[blue edge] (D2) -- (Y2) -- (D3);

        \coordinate (Z) at (8.375, -7.5);
        \draw[red edge] (Y1) -- (Y2); %
        \draw[red edge] (Y1) -- (Z) -- (Y2); %

        \coordinate (S1) at (0.75, -8.125);   
        \coordinate (S2) at (4.875, -11.625); 
        \draw[black edge] (Y1) -- (S1) -- (S2) -- (Z);

        \foreach \p in {A0, A1, A2, A3, A4, B0, B1, B2, B3, C0, C1, C2, C4, D0, D1, D2, D3, S1, S2} {
            \node[vertex] at (\p) {};
        }
        \node[vertex, label={[xshift=-8pt, yshift=-12pt]left:$w_{{2i-1}}$}] at (Y1) {};
        \node[vertex, label={[xshift=8pt, yshift=-12pt]right:$w_{{2i}}$}] at (Y2) {};
        \node[vertex, label=right:$z_i$] at (Z) {};

        \node[fill=white, inner sep=1pt] at (0, 0.2) {\Large $A_{4i-3}$};
        \node[fill=white, inner sep=1pt] at (4, 0.2) {\Large $A_{4i-2}$};
        \node[fill=white, inner sep=1pt] at (8, 0.2) {\Large $A_{4i-1}$};
        \node[fill=white, inner sep=1pt] at (12, 0.2) {\Large $A_{4i}$};
        
        \node at (-3.5, 0.2) {\Huge $\cdots$};
        \node at (15.5, 0.2) {\Huge $\cdots$};

    \end{tikzpicture}        
    \end{center}
     \caption{Extended variable gadget structure. %
     The four displayed cliques form segment $S_i$.
     The attachment triangles $Q_{2i-1}$ and $Q_{2i}$ are based on the edges $e_{4i-2}\in E(A_{4i-2})$ and $e_{4i}\in E(A_{4i})$, respectively.
     Their tips $w_{2i-1}$ and $w_{2i}$, together with $z_i$, form the central triangle $P_i$.
     The colors illustrate the prescribed coloring of a Red variable.
     The three black edges, together with the clause edge $w_{2i-1}z_i$, indicate the corresponding clause gadget.}
    \label{fig:var_with_clause}
\end{figure}
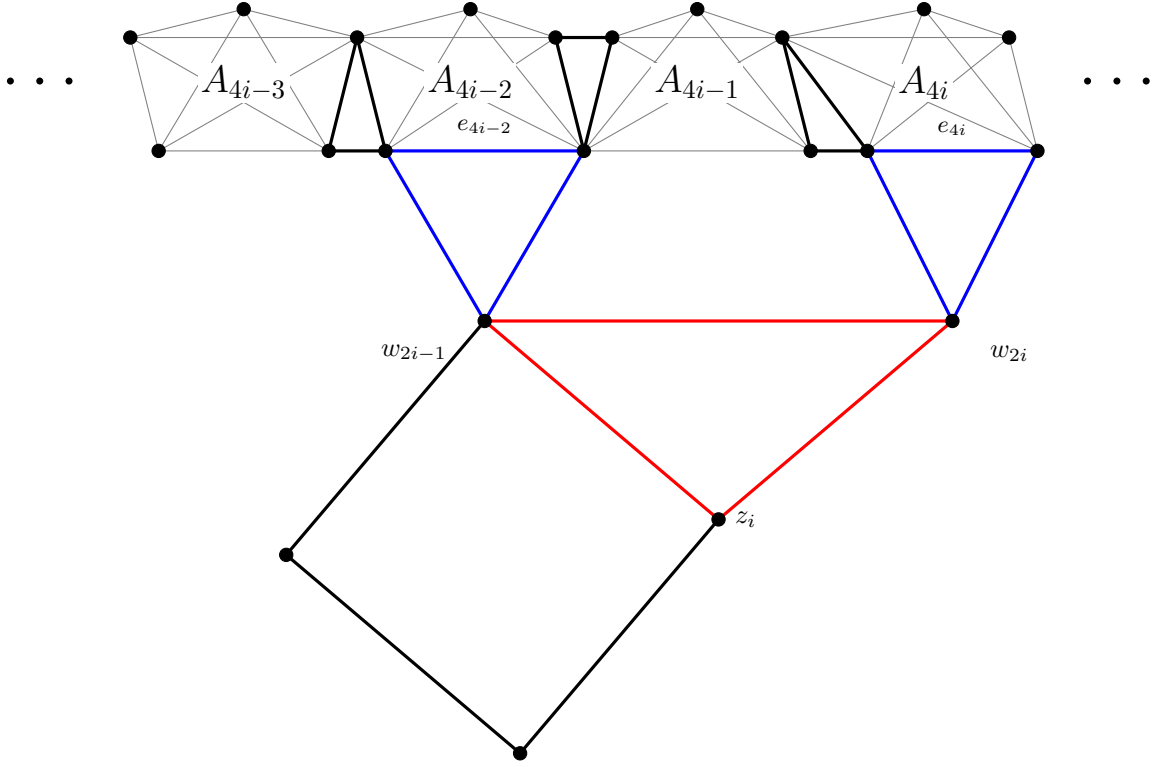

\medskip
\noindent\textbf{Clause Gadget.~~}
Let $C\in\Clauses$ be a clause, and fix an arbitrary cyclic ordering $(V_1,V_2,V_3,V_4)$ of its four variables.
For each $r\in[4]$, let $i_r$ be the unique index such that $\rho_{V_r}(i_r)=(V_r,C)$.
To simplify notation within this clause, write
\[
a_r\coloneqq w^{V_r}_{2i_r-1},
\qquad
b_r\coloneqq w^{V_r}_{2i_r},
\qquad
z_r\coloneqq z^{V_r}_{i_r}.
\]
Thus, $a_rz_r$ is the clause edge associated with the occurrence of $V_r$ in $C$, and $P_{i_r}^{V_r}=a_rb_rz_r$ is its central triangle.
We identify
\[
z_r=a_{r+1}
\qquad\text{for every }r\in[4],
\]
where indices are taken modulo $4$, so that $a_5=a_1$.
Consequently, the four clause edges become $a_1a_2$, $a_2a_3$, $a_3a_4$, $a_4a_1$ and form the clause gadget, which is a copy of $C_4$.
No other identifications are made between distinct variable gadgets.
Every occurrence is assigned to exactly one segment, so each vertex denoted by $a_r$ or $z_r$ participates in exactly one identification, and no vertex denoted by $b_r$ is identified with another vertex.
Moreover, the two vertices in each identified pair initially belong to different variable gadgets, because every clause contains four distinct variables.
Thus no identification creates a loop.
After the identifications, the central triangle associated with $V_r$ is $P_{i_r}^{V_r}=a_rb_ra_{r+1}$, where $b_r$ is not used elsewhere in the clause gadget, and the only edges with both endpoints among $a_1,a_2,a_3,a_4$ are the four distinct clause edges.
All remaining edges retain an endpoint private to their original variable gadget, so no two edges become parallel.
Consequently, the construction produces a simple graph, and every clause edge belongs to exactly one clause gadget.
See \cref{fig:var_with_clause} for the connection between a variable segment and its clause gadget.

\medskip
\noindent\textbf{Constructed Graph.~~}
For every variable $V$, we construct the variable gadget above with $\ell=|O_V|$, initially taking the gadgets for distinct variables to be vertex-disjoint.
We then perform the identifications prescribed by the clause gadgets.
A variable gadget has $O(\ell+1)$ vertices and edges, and the clause identifications introduce no new vertices or edges and only reduce the number of vertices.
Since $\sum_{V\in\Vars}|O_V|=4|\Clauses|$, the resulting simple graph has $O(|\Vars|+|\Clauses|)$ vertices and edges and can be constructed in polynomial time.

\subsection{Properties of Gadgets}

In this section, we describe some basic properties of the constructed gadgets.

\begin{lemma}[Monochromatic Inside Triangle]\label[lemma]{lem:triangle}
Let $G$ be the constructed gadget graph.
Each inside triangle must be monochromatic in any admissible coloring of $G$.
\end{lemma}
\begin{proof}
Let $A$ and $B$ be two copies of $K_5$ sharing exactly one vertex $v$. Let $T=(u,v,w)$ be a triangle such that $uv \in E(A)$, $vw \in E(B)$.
Observe that $T$ is an inside triangle in a variable gadget.

By \cref{lem:bull}, the Red and Blue subgraphs within both $A$ and $B$ are connected, spanning bull graphs. We will show that any non-monochromatic coloring of $T$ creates a cycle of length $\ge 4$.

Assume for contradiction that the edges $uv$ and $vw$ receive different colors. Without loss of generality, let $uv$ be Red and $vw$ be Blue. We consider the color of the edge $uw$:
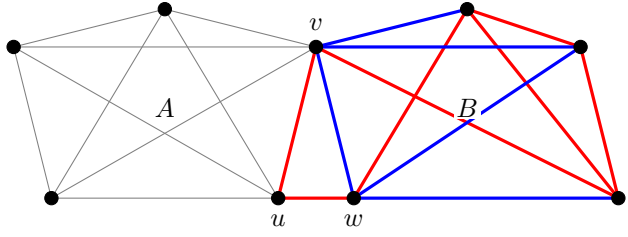
\begin{figure}[t]
    \centering
    \begin{tikzpicture}

        \coordinate (A0) at (0, 1.5);
        \coordinate (A1) at (2, 1);      %
        \coordinate (A2) at (1.5, -1);   %
        \coordinate (A3) at (-1.5, -1);  %
        \coordinate (A4) at (-2, 1);     %

        \coordinate (B0) at (4, 1.5);
        \coordinate (B1) at (5.5, 1);    
        \coordinate (B2) at (6, -1);     
        \coordinate (B3) at (2.5, -1);   
        \coordinate (B4) at (2, 1);      %

        \begin{scope}[gray, thin]
            \draw (A0)--(A1) (A0)--(A2) (A0)--(A3) (A0)--(A4);
            \draw (A1)--(A2) (A1)--(A3) (A1)--(A4);
            \draw (A2)--(A3) (A2)--(A4);
            \draw (A3)--(A4);
        \end{scope}

        \draw[red edge] (B1) -- (B0);
        \draw[red edge] (B1) -- (B2);
        \draw[red edge] (B0) -- (B2);
        \draw[red edge] (B0) -- (B3);
        \draw[red edge] (B4) -- (B2);

        \draw[blue edge] (B0) -- (B4);
        \draw[blue edge] (B1) -- (B4);
        \draw[blue edge] (B1) -- (B3);
        \draw[blue edge] (B2) -- (B3);

        \draw[red edge] (A2) -- (A1);   %
        \draw[blue edge] (A1) -- (B3); %
        \draw[red edge] (A2) -- (B3);  %

        \node[vertex] at (A0) {}; 
        \node[vertex, label=above:$v$] at (A1) {};
        \node[vertex, label=below:$u$] at (A2) {}; 
        \node[vertex] at (A3) {}; 
        \node[vertex] at (A4) {};
        
        \node[vertex] at (B0) {}; 
        \node[vertex] at (B1) {};
        \node[vertex] at (B2) {}; 
        \node[vertex, label=below:$w$] at (B3) {};

        \node[fill=white, inner sep=1pt] at (0,0.2) {$A$};
        \node[fill=white, inner sep=1pt] at (4,0.2) {$B$};
        
    \end{tikzpicture}
    \caption{Illustration for the Monochromatic Inside Triangle Lemma (\cref{lem:triangle}). The highlighted Red $C_5$ demonstrates the contradiction arising when one edge of the inside triangle is Blue and the other two are Red.}
    \label{fig:lemma_triangle_case1}
\end{figure}

\begin{itemize}
    \item \textbf{Case 1: $uw$ is Red.} 
Since the Red bull in $B$ is connected and spanning, it contains a simple Red path from $v$ to $w$.
    Because the edge $vw$ is Blue, this path has length at least $2$.
    Its union with the Red edges $uv$ and $uw$ is a simple Red cycle of length at least $4$ ($u-v-\dots-w-u$), a contradiction; this situation is illustrated in \cref{fig:lemma_triangle_case1}. 
    \item \textbf{Case 2: $uw$ is Blue.} By symmetry, the Blue bull in $A$ contains a simple Blue path from $u$ to $v$ of length at least $2$, because the direct edge $uv$ is Red.
Its union with the Blue edges $uw$ and $vw$ is a simple Blue cycle of length at least $4$, a contradiction.
\end{itemize}
Because both cases result in a forbidden $C_{\ge 4}$, the edges $uv$ and $vw$ must have the same color.
By exchanging Red and Blue if necessary, assume that they are both Red.
Suppose for contradiction that $uw$ is Blue.
Since the Blue subgraph of $A$ is connected and spanning, it contains a simple Blue path $R_A$ from $u$ to $v$.
This path has length at least $2$, because $uv$ is Red.
Similarly, the Blue subgraph of $B$ contains a simple Blue path $R_B$ from $v$ to $w$ of length at least $2$.
Since $A\cap B=\{v\}$, the paths $R_A$ and $R_B$ meet exactly at $v$.
The union of $R_A$, $R_B$, and the Blue edge $uw$ is a Blue cycle of length at least $5$, a contradiction.
Consequently, $uw$ is Red, and $T$ is monochromatic.
\end{proof}

\begin{lemma}[Variable Coloring]\label[lemma]{lem:variable}
Let $V$ be a variable gadget of length $\ell$.
In every admissible coloring, all inside triangles of $V$ have a common color.
Moreover, for every $i\in[4\ell]$, the edges $v_iu_i\in E(A_i)$ and $v_iu'_i\in E(A_{i+1})$ are horns of that color in the respective bull decompositions.
\end{lemma}
We denote the common color of the inside triangles of $V$ by $c(V)$ and the other color by $\neg{c(V)}$.
The color propagation described by \cref{lem:variable} is illustrated in \cref{fig:valid_variable_coloring}.

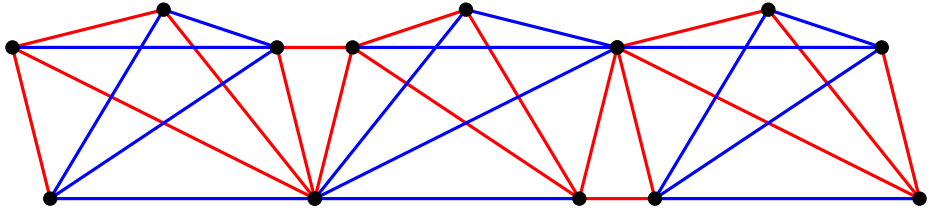
\begin{figure}[b]
    \centering
    \begin{tikzpicture}
        \coordinate (V0_1) at (4, 1.5);
        \coordinate (V1_1) at (5.5, 1);
        \coordinate (V2_1) at (6, -1);
        \coordinate (V3_1) at (2.5, -1);
        \coordinate (V4_1) at (2, 1);

        \coordinate (V0_2) at (8, 1.5);
        \coordinate (V1_2) at (10, 1);
        \coordinate (V2_2) at (9.5, -1);
        \coordinate (V3_2) at (6, -1);
        \coordinate (V4_2) at (6.5, 1);

        \coordinate (V0_3) at (12, 1.5);
        \coordinate (V1_3) at (13.5, 1);
        \coordinate (V2_3) at (14, -1);
        \coordinate (V3_3) at (10.5, -1);
        \coordinate (V4_3) at (10, 1);

        \foreach \i in {1,2,3} {
            \draw[red edge] (V0_\i) -- (V4_\i) -- (V2_\i) -- cycle;
            \draw[red edge] (V4_\i) -- (V3_\i);
            \draw[red edge] (V2_\i) -- (V1_\i);

            \draw[blue edge] (V0_\i) -- (V1_\i) -- (V3_\i) -- cycle;
            \draw[blue edge] (V1_\i) -- (V4_\i); 
            \draw[blue edge] (V3_\i) -- (V2_\i); 
        }

        \draw[red edge] (V1_1) -- (V4_2);
        \draw[red edge] (V2_2) -- (V3_3);

        \foreach \i in {1,2,3} {
            \node[vertex] at (V0_\i) {};
            \node[vertex] at (V1_\i) {};
            \node[vertex] at (V2_\i) {};
            \node[vertex] at (V3_\i) {};
            \node[vertex] at (V4_\i) {};
        }

    \end{tikzpicture}
    \caption{An admissible coloring of part of the core of a Red variable gadget, in which each displayed bottom edge is a Blue horn in its clique's bull decomposition.}
    \label{fig:valid_variable_coloring}
\end{figure}

\begin{proof}[Proof of \cref{lem:variable}]
By \cref{lem:triangle}, every inside triangle is monochromatic.
We first prove the horn assertion.
Fix $i\in[4\ell]$.
By exchanging Red and Blue if necessary, assume that $T_i$ is Red.
Suppose that $v_iu_i$ belongs to the central triangle of the Red bull in $A_i$.
Then some vertex $x\in V(A_i)$ forms a Red triangle with $v_i$ and $u_i$.
The vertices $x,v_i,u_i$ are distinct, while $u'_i\in V(A_{i+1})\setminus\{v_i\}$ and $A_i\cap A_{i+1}=\{v_i\}$, so $x,v_i,u'_i,u_i$ are pairwise distinct.
Together with the Red inside triangle $T_i$, this gives the Red cycle $xv_iu'_iu_ix$ of length $4$, a contradiction.
Thus, $v_iu_i$ is a Red horn in $A_i$.
The same argument applied in $A_{i+1}$ shows that $v_iu'_i$ is a Red horn there.
Since $i$ was arbitrary and the argument is symmetric in the two colors, the horn assertion follows.

It remains to show that all inside triangles have the same color.
Otherwise, for some $j\in[4\ell-1]$, the consecutive inside triangles $T_j$ and $T_{j+1}$ have different colors.
By exchanging Red and Blue if necessary, assume that $T_j$ is Red and $T_{j+1}$ is Blue.

By the horn assertion proved above, $v_ju'_j$ is a Red horn in $A_{j+1}$, whereas $v_{j+1}u_{j+1}$ is a Blue horn in $A_{j+1}$.
By construction, the four vertices $v_j,u'_j,v_{j+1},u_{j+1}$ are pairwise distinct, so these two horn edges are vertex-disjoint.
This contradicts \cref{lem:bull}, according to which every Red horn has exactly one common endpoint with every Blue horn.
Therefore, all inside triangles have the same color.
\end{proof}

\begin{lemma}[Attachment $\neg{c(V)}$-colored Path]\label[lemma]{lem:ClauseTriangles}
Let $G$ be the graph produced by the entire construction.
For a segment $S_i$ of a variable gadget $V$, let $H_i$ be the subgraph consisting of the cliques $A_{4i-3},\ldots,A_{4i}$, the inside triangles joining consecutive cliques among them, and the attachment triangles $Q_{2i-1}$ and $Q_{2i}$, but no edge of the central triangle $P_i$.
In every admissible coloring of $G$, the graph $H_i$ contains a $\neg{c(V)}$-colored path of length at least $4$ between $w_{2i-1}$ and $w_{2i}$.
\end{lemma}

\begin{proof}
Fix a segment $S_i$ of a variable gadget $V$.
We first show that each of its tips has a $\neg{c(V)}$-colored edge to the core of the segment.

Let $w$ be one of the tips, and let $e_j=xy$ be the edge on which its attachment triangle is built; thus, $j\in\{4i-2,4i\}$.
The clique $A_j$ contains the two clique-edges $v_{j-1}u'_{j-1}$ and $v_ju_j$ of the adjacent inside triangles.
By \cref{lem:variable}, both are $c(V)$-colored horns.
They are disjoint, and a bull graph has exactly two horns, so these are precisely the two $c(V)$-colored horns in $A_j$.
Since $e_j=u'_{j-1}v_j$ is distinct from both of them, $e_j$ is not a $c(V)$-colored horn.

Suppose for contradiction that both $wx$ and $wy$ have color $c(V)$.
We claim that $A_j$ contains a $c(V)$-colored $xy$ path avoiding $xy$.
If $xy$ has color $\neg{c(V)}$, this follows from the fact that the $c(V)$-colored bull in $A_j$ is connected and spanning.
If $xy$ has color $c(V)$, then, since it is not a horn, it belongs to the central triangle of that bull, whose other two edges form such a path.
In either case, this path has length at least $2$.
Together with the path $xwy$, it forms a $c(V)$-colored cycle of length at least $4$, a contradiction.
Hence, at least one of $wx$ and $wy$ has color $\neg{c(V)}$.

Choose a $\neg{c(V)}$-colored edge from $w_{2i-1}$ to a vertex $p\in A_{4i-2}$ and one from $w_{2i}$ to a vertex $q\in A_{4i}$.
Set $a\coloneqq 4i-2$.
The $\neg{c(V)}$-colored bull in $A_a$ contains a path from $p$ to $v_a$, the one in $A_{a+1}$ contains a path from $v_a$ to $v_{a+1}$, and the one in $A_{a+2}$ contains a path from $v_{a+1}$ to $q$.
Because consecutive cliques intersect only in their specified shared vertex, these paths, together with the two chosen attachment edges, concatenate to a $\neg{c(V)}$-colored path in $H_i$ from $w_{2i-1}$ to $w_{2i}$.

This path has length at least $4$.
Indeed, its two attachment edges contribute $2$.
Moreover, $v_a\ne v_{a+1}$, so the path inside $A_{a+1}$ has positive length.
Finally, $q$ is an endpoint of $e_{a+2}=u'_{a+1}v_{a+2}$, and both endpoints of this edge are distinct from $v_{a+1}$.
Hence, the path inside $A_{a+2}$ also has positive length.
\end{proof}

\begin{lemma}[$c(V)$-colored Path in a Central Triangle]\label[lemma]{lem:centraltriag}
Let $V$ be a variable gadget of length $\ell$.
In every admissible coloring and for every $i\in[\ell]$, the edge $w_{2i-1}w_{2i}$ has color $c(V)$, and at least one of the edges $w_{2i-1}z_i$ and $w_{2i}z_i$ has color $c(V)$.
Consequently, the central triangle $P_i$ contains a $c(V)$-colored path of length at most $2$ from $w_{2i-1}$ to $z_i$.
\end{lemma}
\begin{proof}
By \cref{lem:ClauseTriangles}, $H_i$ contains a $\neg{c(V)}$-colored path $R$ of length at least $4$ from $w_{2i-1}$ to $w_{2i}$.
By the definition of $H_i$, no edge of $H_i$ is incident with $z_i$.
Indeed, before the clause identifications, $z_i$ lies outside $H_i$.
The identification pairs are vertex-disjoint, and the tip paired with $z_i$ belongs to a different variable gadget because every clause contains four distinct variables.
Consequently, no edge of $H_i$ becomes incident with $z_i$.
In particular, $z_i\notin V(R)$, and the path $R$ uses no edge of $P_i$.
If $w_{2i-1}w_{2i}$ also had color $\neg{c(V)}$, then $R\cup\{w_{2i-1}w_{2i}\}$ would contain a $\neg{c(V)}$-colored cycle of length at least $5$, a contradiction.
Thus, $w_{2i-1}w_{2i}$ has color $c(V)$.

If both $w_{2i-1}z_i$ and $w_{2i}z_i$ had color $\neg{c(V)}$, then these two edges together with $R$ would form a $\neg{c(V)}$-colored cycle of length at least $6$, again a contradiction.
Hence, at least one of these two edges has color $c(V)$.
If $w_{2i-1}z_i$ has color $c(V)$, it is the required path; otherwise, $w_{2i-1}w_{2i}z_i$ is the required $c(V)$-colored path of length $2$.
\end{proof}

\subsection{Proof of \cref{thm:main}}

\noindent\textbf{Admissible Coloring to NAE Valuation.~~}
Let $\varphi=(\Vars,\Clauses)$ be an instance of \textsc{Monotone NAE-4-SAT}, and let $G_\varphi$ be the graph constructed above.
Suppose that $G_\varphi$ has an admissible coloring.
For each variable $V\in\Vars$, set
\[
\sigma(V)=\begin{cases}
\top,&\text{if }c(V)=\text{Red},\\
\bot,&\text{if }c(V)=\text{Blue}.
\end{cases}
\]
We claim that $\sigma$ is an NAE-satisfying valuation.

Suppose to the contrary that all four variables of some clause $C$ receive the same truth value under $\sigma$, and let $(V_1,V_2,V_3,V_4)$ be the cyclic ordering fixed for $C$ in the construction.
Then $c(V_1)=c(V_2)=c(V_3)=c(V_4)$.
By exchanging the names Red and Blue if necessary, assume that this common color is Red.
We use the notation $a_r,b_r,z_r$ from the construction of the clause gadget, so $z_r=a_{r+1}$ with indices modulo $4$.
By \cref{lem:centraltriag}, for every $r\in[4]$, the central triangle $P_{i_r}^{V_r}$ associated with the occurrence of $V_r$ contains a Red path $R_r$ from $a_r$ to $z_r=a_{r+1}$; this path is either the clause edge $a_rz_r$ or the two-edge path $a_rb_rz_r$.
Because the four variables of $C$ are distinct and no other identifications are made between their gadgets, the vertices $a_1,\ldots,a_4,b_1,\ldots,b_4$ are pairwise distinct.
Hence, with indices modulo $4$, consecutive paths $R_r$ and $R_{r+1}$ meet exactly at $a_{r+1}$, while nonconsecutive paths are vertex-disjoint.
Their union is therefore a Red subdivision of $C_4$, hence a monochromatic cycle of length at least $4$.
This contradicts admissibility.
Thus every clause contains variables of both truth values, and $\sigma$ NAE-satisfies $\varphi$.

\bigskip

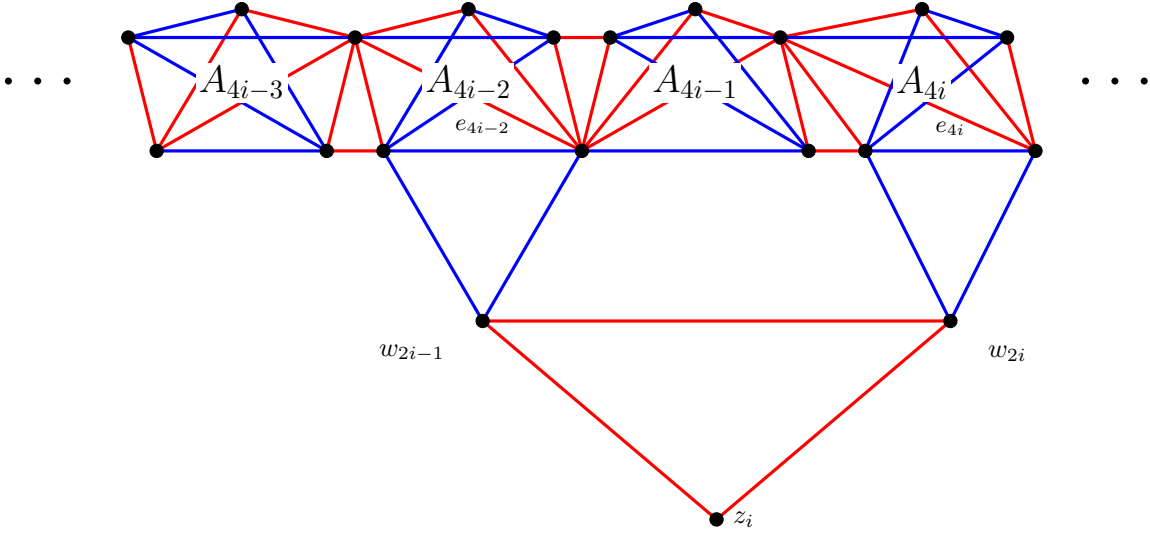
\begin{figure}[t]
    \begin{center}
    \hspace*{-0.2cm}
    \begin{tikzpicture}[
        scale=0.75
    ]
        \coordinate (A0) at (0, 1.5);
        \coordinate (A1) at (2, 1);
        \coordinate (A2) at (1.5, -1);
        \coordinate (A3) at (-1.5, -1);
        \coordinate (A4) at (-2, 1);

        \coordinate (B0) at (4, 1.5);
        \coordinate (B1) at (5.5, 1);
        \coordinate (B2) at (6, -1);
        \coordinate (B3) at (2.5, -1);
        \coordinate (C0) at (8, 1.5);
        \coordinate (C1) at (9.5, 1);
        \coordinate (C2) at (10, -1);
        \coordinate (C4) at (6.5, 1);

        \coordinate (D0) at (12, 1.5);
        \coordinate (D1) at (13.5, 1);
        \coordinate (D2) at (14, -1);
        \coordinate (D3) at (11, -1);
        \coordinate (Y1) at (4.25, -4);
        \coordinate (Y2) at (12.5, -4);
        \coordinate (Z) at (8.375, -7.5);

        \draw[red edge] (A0) -- (A1) -- (A3) -- cycle;
        \draw[red edge] (A3) -- (A4);
        \draw[red edge] (A1) -- (A2);
        \draw[blue edge] (A0) -- (A2) -- (A4) -- cycle;
        \draw[blue edge] (A4) -- (A1);
        \draw[blue edge] (A2) -- (A3);

        \draw[red edge] (B0) -- (A1) -- (B2) -- cycle;
        \draw[red edge] (A1) -- (B3);
        \draw[red edge] (B2) -- (B1);
        \draw[blue edge] (B0) -- (B1) -- (B3) -- cycle;
        \draw[blue edge] (B1) -- (A1);

        \draw[red edge] (C0) -- (B2) -- (C1) -- cycle;
        \draw[red edge] (B2) -- (C4);
        \draw[red edge] (C1) -- (C2);
        \draw[blue edge] (C0) -- (C2) -- (C4) -- cycle;
        \draw[blue edge] (C4) -- (C1);
        \draw[blue edge] (C2) -- (B2);

        \draw[red edge] (D0) -- (C1) -- (D2) -- cycle;
        \draw[red edge] (C1) -- (D3);
        \draw[red edge] (D2) -- (D1);
        \draw[blue edge] (D0) -- (D1) -- (D3) -- cycle;
        \draw[blue edge] (D1) -- (C1);

        \draw[red edge] (A2) -- (B3);
        \draw[red edge] (B1) -- (C4);
        \draw[red edge] (C2) -- (D3);

        \draw[blue edge] (B2) -- node[above, text=black, font=\small, yshift=2pt] {$e_{4i-2}$} (B3);
        \draw[blue edge] (B2) -- (Y1) -- (B3);
        \draw[blue edge] (D2) -- node[above, text=black, font=\small, yshift=2pt] {$e_{4i}$} (D3);
        \draw[blue edge] (D2) -- (Y2) -- (D3);

        \draw[red edge] (Y1) -- (Y2);
        \draw[red edge] (Y1) -- (Z) -- (Y2);

        \foreach \p in {A0, A1, A2, A3, A4, B0, B1, B2, B3, C0, C1, C2, C4, D0, D1, D2, D3} {
            \node[vertex] at (\p) {};
        }

        \node[vertex, label={[xshift=-8pt, yshift=-12pt]left:$w_{2i-1}$}] at (Y1) {};
        \node[vertex, label={[xshift=8pt, yshift=-12pt]right:$w_{2i}$}] at (Y2) {};
        \node[vertex, label=right:$z_i$] at (Z) {};

        \node[fill=white, inner sep=1pt] at (0, 0.2) {\Large $A_{4i-3}$};
        \node[fill=white, inner sep=1pt] at (4, 0.2) {\Large $A_{4i-2}$};
        \node[fill=white, inner sep=1pt] at (8, 0.2) {\Large $A_{4i-1}$};
        \node[fill=white, inner sep=1pt] at (12, 0.2) {\Large $A_{4i}$};

        \node at (-3.5, 0.2) {\Huge $\cdots$};
        \node at (15.5, 0.2) {\Huge $\cdots$};

    \end{tikzpicture}
    \end{center}
    \caption{The prescribed coloring of a full segment of a Red variable gadget.}
    \label{fig:valid_full_variable}
\end{figure}

\noindent\textbf{NAE Valuation to Admissible Coloring.~~}
Conversely, suppose that $\sigma$ is an NAE-satisfying valuation of $\varphi$.
For each variable $V$, set $c(V)=\text{Red}$ if $\sigma(V)=\top$ and $c(V)=\text{Blue}$ otherwise.
We prescribe an edge-coloring as follows; \cref{fig:valid_full_variable} illustrates one segment when $c(V)=\text{Red}$.
Fix a variable gadget $V$ of length $\ell$.
For each internal clique $A_j$, where $2\leq j\leq 4\ell$, let $x_j$ be the unique vertex of $A_j$ outside $\{v_{j-1},u'_{j-1},v_j,u_j\}$.
Within $A_j$, color the triangle $v_{j-1}v_jx_j$ and the two horns $v_{j-1}u'_{j-1}$, $v_ju_j$ with color $c(V)$.
Color the five complementary edges of $A_j$ with color $\neg{c(V)}$.
These form the complementary bull, whose horns are $e_j=u'_{j-1}v_j$ and $u_jv_{j-1}$.
For each of the boundary cliques $A_1$ and $A_{4\ell+1}$, choose complementary bull colorings in which the unique clique-edge belonging to its adjacent inside triangle is a $c(V)$-colored horn.
Such a coloring exists because, by relabeling a complementary bull decomposition of $K_5$, any prescribed edge can be made a horn of either chosen color.
Color every cross edge $u_ju'_j$ with $c(V)$, so every inside triangle is monochromatic in color $c(V)$.
For each segment $S_i$, color the two remaining edges of each of its attachment triangles $Q_{2i-1}$ and $Q_{2i}$ with $\neg{c(V)}$, and color every edge of its central triangle $P_i$ with $c(V)$; the base edges $e_{4i-2}$ and $e_{4i}$ already have color $\neg{c(V)}$.
This specifies the color of every edge of $G_\varphi$.
The local variable pattern appears in \cref{fig:valid_full_variable}, and \cref{fig:exact_color_structure} shows how the prescribed colors meet around a clause gadget.

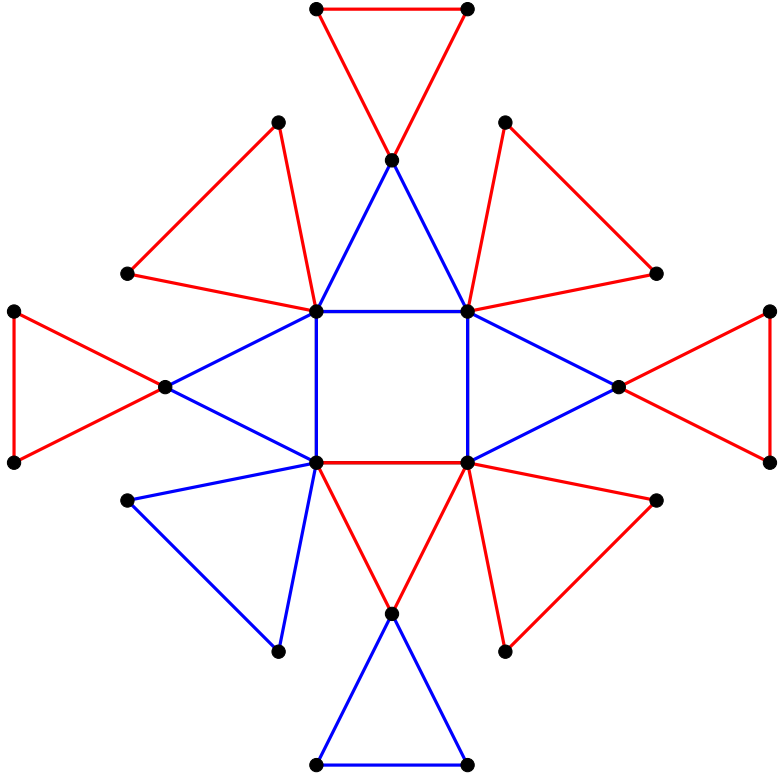
\begin{figure}[t]
    \centering
    \begin{tikzpicture}

        \coordinate (VTL) at (-1, 1);   %
        \coordinate (VTR) at (1, 1);    %
        \coordinate (VBR) at (1, -1);   %
        \coordinate (VBL) at (-1, -1);  %

        \coordinate (TT) at (0, 3);     %
        \coordinate (TR) at (3, 0);     %
        \coordinate (TB) at (0, -3);    %
        \coordinate (TL) at (-3, 0);    %

        \coordinate (OTT_L) at (-1, 5); %
        \coordinate (OTT_R) at (1, 5);  %
        
        \coordinate (ORT_T) at (5, 1);  %
        \coordinate (ORT_B) at (5, -1); %
        
        \coordinate (OBT_L) at (-1, -5);%
        \coordinate (OBT_R) at (1, -5); %
        
        \coordinate (OLT_T) at (-5, 1); %
        \coordinate (OLT_B) at (-5, -1);%

        \coordinate (OTL_T) at (-1.5, 3.5); %
        \coordinate (OTL_L) at (-3.5, 1.5); %
        
        \coordinate (OTR_T) at (1.5, 3.5);  %
        \coordinate (OTR_R) at (3.5, 1.5);  %
        
        \coordinate (OBR_B) at (1.5, -3.5); %
        \coordinate (OBR_R) at (3.5, -1.5); %
        
        \coordinate (OBL_B) at (-1.5, -3.5);%
        \coordinate (OBL_L) at (-3.5, -1.5);%

        \draw[black edge] (VTL) -- (VTR) -- (VBR) -- (VBL) -- cycle;

        \draw[blue edge] (VTL) -- (TT) -- (VTR) -- cycle; %
        \draw[blue edge] (VTR) -- (TR) -- (VBR) -- cycle; %
        \draw[red edge]  (VBR) -- (TB) -- (VBL) -- cycle; %
        \draw[blue edge] (VBL) -- (TL) -- (VTL) -- cycle; %

        \draw[red edge]  (TT) -- (OTT_L) -- (OTT_R) -- cycle; %
        \draw[red edge]  (TR) -- (ORT_T) -- (ORT_B) -- cycle; %
        \draw[blue edge] (TB) -- (OBT_L) -- (OBT_R) -- cycle; %
        \draw[red edge]  (TL) -- (OLT_T) -- (OLT_B) -- cycle; %

        \draw[red edge]  (VTL) -- (OTL_T) -- (OTL_L) -- cycle; %
        \draw[red edge]  (VTR) -- (OTR_T) -- (OTR_R) -- cycle; %
        \draw[red edge]  (VBR) -- (OBR_B) -- (OBR_R) -- cycle; %
        \draw[blue edge] (VBL) -- (OBL_B) -- (OBL_L) -- cycle; %

        \foreach \pos in {VTL, VTR, VBR, VBL, 
                          TT, TR, TB, TL, 
                          OTT_L, OTT_R, ORT_T, ORT_B, OBT_L, OBT_R, OLT_T, OLT_B, 
                          OTL_T, OTL_L, OTR_T, OTR_R, OBR_B, OBR_R, OBL_B, OBL_L}
            \node[vertex] at (\pos) {};

    \end{tikzpicture}
    \caption{The prescribed coloring around a clause containing three Blue variables and one Red variable. The central triangles have the colors of their variables, while the corresponding attachment triangles have the opposite colors.}
    \label{fig:exact_color_structure}
\end{figure}

We now prove that the prescribed coloring is admissible.
By symmetry between Red and Blue, it suffices to prove that the subgraph formed by the Red edges is a triangular forest.
Let $G_{\mathrm{Red}}$ be this subgraph, and let $B_{\mathrm{Red}}$ be obtained from $G_{\mathrm{Red}}$ by deleting all edges of $P_i^V$ for every $V\in\Vars$ and $i\in[|O_V|]$, and then omitting isolated vertices.
We first show that $B_{\mathrm{Red}}$ is a triangular forest.

Fix a variable $V$.
Suppose first that $c(V)=\text{Red}$.
Within every clique $A_j$, the Red edges induce a bull, which consists of one triangular block and two pendant edges.
Before adding the cross edges, the Red bulls in consecutive cliques meet only at the corresponding shared vertex $v_j$.
Their union is therefore a triangular forest.
For each $j\in[4\ell]$, the two clique-edges of the inside triangle $T_j$ are Red horns in their respective bulls; this follows from the explicit prescription for the internal cliques and from the choice of the bull colorings in the two boundary cliques.
The Red cross edge $u_ju'_j$ completes these two horn edges to the triangle $T_j$.
Consequently, the nontrivial blocks of the resulting Red core are precisely the triangle blocks of the Red bulls and the inside triangles $T_j$.
Any unused horn in a boundary clique remains a bridge.
The pairwise-distinct condition in the internal cliques ensures that consecutive inside triangles are vertex-disjoint.
Thus, the Red core of $V$ is a triangular forest.
The attachment triangles of $V$ are Blue, so they contribute no edge to $B_{\mathrm{Red}}$.

Now suppose that $c(V)=\text{Blue}$.
No cross edge of an inside triangle then belongs to $B_{\mathrm{Red}}$.
The Red bulls in consecutive cliques meet only at their shared vertex, and hence their union is a triangular forest.
Moreover, every attachment triangle of $V$ is Red.
Its base edge is one of the edges $e_{4i-2}$ or $e_{4i}$, which is a horn, and hence a bridge, of the corresponding Red bull.
Adding the two Red edges incident with the fresh tip turns this bridge into the triangular block $Q_{2i-1}$ or $Q_{2i}$ and creates no other block.
It follows that the Red core together with all Red attachment triangles of $V$ is a triangular forest.

The only identifications between distinct variable gadgets identify a vertex of the form $z_i^V$ with a tip of another variable gadget.
After the edges of the central triangles are deleted, $z_i^V$ has no incident edge contributed by the gadget of $V$.
Thus, these identifications do not join two edge-containing variable pieces of $B_{\mathrm{Red}}$.
It follows that $B_{\mathrm{Red}}$ is a disjoint union of triangular forests and is therefore itself a triangular forest.

It remains to add the Red central triangles.
Fix a clause $C$ and use the notation $a_r,b_r,z_r,V_r,i_r$ from its construction.
For every $r\in[4]$, set
\[
P_r\coloneqq P_{i_r}^{V_r}=a_rb_ra_{r+1}.
\]
For every $r\in[4]$, the triangle $P_r$ is Red exactly when $c(V_r)=\text{Red}$.
Since $\sigma$ NAE-satisfies $C$, at least one variable of $C$ is Red and at least one is Blue.
Therefore, the Red indices decompose into nonempty proper maximal intervals in the cyclic ordering of the clause.

Let $I$ be one such interval, and let $p$ and $q$ be its first and last indices in cyclic order.
The graph $K_I\coloneqq\bigcup_{r\in I}P_r$ is a chain of triangles: consecutive triangles share exactly the vertex $a_{r+1}$, and nonconsecutive triangles are vertex-disjoint.
In particular, $K_I$ is a triangular forest.
No vertex of $K_I$ other than its terminal vertex $a_{q+1}$ is incident with a Red edge of $B_{\mathrm{Red}}$.
Indeed, the vertices $a_r$ and $b_r$ for $r\in I$ are tips of attachment triangles belonging to Red variables, so those attachment triangles are Blue.
On the other hand, maximality of $I$ gives $c(V_{q+1})=\text{Blue}$.
Hence, the attachment triangle incident with the tip $a_{q+1}$ is Red and belongs to $B_{\mathrm{Red}}$.
Thus, $K_I$ meets $B_{\mathrm{Red}}$ in exactly the single vertex $a_{q+1}$.

Distinct maximal intervals give disjoint triangle chains outside $B_{\mathrm{Red}}$, and every occurrence belongs to exactly one clause.
Consequently, $G_{\mathrm{Red}}$ is obtained from $B_{\mathrm{Red}}$ by repeatedly taking 1-sums with the triangular forests $K_I$.
Since triangular forests are closed under 1-sums, $G_{\mathrm{Red}}$ is a triangular forest.
The same argument, with Red and Blue interchanged, proves that the Blue subgraph is also a triangular forest.
Therefore, the prescribed coloring is admissible.

\bigskip
Together with the reverse direction, this shows that $\varphi$ is NAE-satisfiable if and only if $G_\varphi$ admits a decomposition of its edge set into two triangular forests.
The construction is computable in polynomial time.
In particular, it uses
\[
\sum_{V\in\Vars}\bigl(4|O_V|+1\bigr)=16|\Clauses|+|\Vars|
\]
copies of $K_5$ and only linearly many additional vertices and edges.
Hence this is a polynomial-time many-one reduction.

It remains to verify the maximum-degree bound.
A shared vertex $v_j$ has four neighbors in $A_j$ and four neighbors in $A_{j+1}$, and these eight neighbors are distinct because $A_j\cap A_{j+1}=\{v_j\}$.
It is not incident with a cross edge.
If $j\in\{4i-2,4i\}$ for some $i$, then $v_j$ is also an endpoint of one attachment base and gains one additional neighbor, namely the corresponding tip.
Thus, every vertex $v_j$ has degree at most $9$.

Every vertex $u_j$ or $u'_j$ has four clique-neighbors and one cross-edge neighbor.
A vertex $u'_j$ may additionally be an endpoint of one attachment base, so these vertices have degree at most $6$.
Every remaining clique vertex has degree $4$.
Before the clause identifications, each tip $w_j$ has degree $4$ and each vertex $z_i$ has degree $2$.
Each clause identification merges a $z$-vertex with a tip from a different variable gadget; their neighborhoods are disjoint, so the resulting vertex has degree $6$.
No other degree is changed by the identifications.
Consequently, the maximum degree of $G_\varphi$ is at most $9$, with degree $9$ attained by the vertices $v_{4i-2}$ and $v_{4i}$.
This concludes the proof of \cref{thm:main}.
\qed{}

\section{Conclusions}

We proved that \textsc{Triangular-Forest Decomposition} is \NP-complete for every fixed $k\ge 2$.
For $k\ge 3$, this follows already from the general theorem of Lee, Liu, and Tsai~\cite{outerthickness}; the contribution of this paper is the case $k=2$.
The most important open problem is whether the same phenomenon holds for other graph classes covered by~\cref{thm:llt}.
In particular, the complexity of deciding whether a graph has outerthickness at most $2$ remains open for general graphs.
Proving that this problem is \NP-hard would, via the parameter-preserving reduction of~\cite[Theorem~6]{BHMOVW24}, also prove that deciding whether the uncrossed number of a graph is at most $2$ is \NP-hard.

\bibliographystyle{alphaurl}
\bibliography{lit}

\end{document}